\documentclass[runningheads,a4paper]{llncs}
\usepackage{epsfig}
\usepackage{amssymb,amsmath}
\usepackage{mathrsfs}
\usepackage{graphicx}
\usepackage{multirow}
\usepackage{color}

\usepackage{url}

\newcommand{\remove}[1]{}

\newcommand{\AMDdelta}{\delta}
\newcommand{\rhoAMD}{$\rho$-AMD }
\newcommand{\LVAMD}{$\rho^{LV}$-AMD }

\newcommand{\sfH}{\mathsf{H}}

\newcommand{\bP}{\mathsf{Pr}}

\newcommand{\bF}{{\mathbb F}}

\newcommand{\SD}{\mathsf{SD}}

\newcommand{\ext}{\mathsf{Ext}}

\newcommand{\bfa}{{\bf a}}

\newcommand{\bfm}{{\bf m}}

\newcommand{\bfs}{{\bf s}}

\newcommand{\bfx}{{\bf x}}

\newcommand{\bfM}{{\bf M}}
\newcommand{\bfR}{{\bf R}}
\newcommand{\bfS}{{\bf S}}

\newcommand{\bfX}{{\bf X}}
\newcommand{\bfZ}{{\bf Z}}

\newcommand{\cS}{{\cal S}}

\newcommand{\eeeenc}{\mathsf{Enc}}
\newcommand{\eeedec}{\mathsf{Dec}}

\newcommand{\share}{\mathsf{Share}}
\newcommand{\recover}{\mathsf{Recover}}
\newcommand{\rsss}{\mathsf{rsss}}
\newcommand{\rrsss}{\mathsf{rrsss}}
\newcommand{\amd}{\mathsf{amd}}

\newtheorem{construction}{Construction}
\begin{document}
\title{Detecting Algebraic Manipulation in Leaky Storage Systems}
\author{Fuchun Lin, Reihaneh Safavi-Naini, Pengwei Wang}
\institute{Department of Computer Science\\University of Calgary, CANADA}                                                                                                                                                                                                                                                                                                                                                                                                                                                                                                                                                                                                                                                                                                                                                                                                                                                                                                                                                                                                                                                                                                                                                                                                                                                                                                                                                                                                                                                                                                                                                                                                                                                                                                                                                                                                                                                                                                                                                                                                                                                                                                                                                                                                                                                                                                                                                                                                                                                                                                                                                                                                                                                                                                                                                                                                                                                                                                                                                                                                                                                                                                                                                                                                                                                                                                                                                                                                                                                                                                                                                                                                                                                                                                                                                                                                                                                                                                                                                                                                                                                                                                                                                                                                                                                                                                                                                                                                                                                                                                                                                                                                                                                                                                                                                                                                                                                                                                                                                                                                                                                                                                                                                                                                                                                                                                                                                                                                                                                                                                                                                                                                                                                                                                                                                                                                                                                                                                                                                                                                                                                                                                                                                                                                                                                                                                                                                                                                                                                                                                                                                                                                                                                                                                                                                                                                                                                                                                                                                                                                                                                                                                                                          

\maketitle

\begin{abstract}
Algebraic Manipulation Detection (AMD) Codes detect   adversarial noise that is
added to a coded message  which is stored in a storage that is opaque to the adversary.
We study AMD codes 
when the  storage can leak up to  $\rho\log|\mathcal{G}|$ bits of information about the stored codeword, where $\mathcal{G}$ is the group 
that contains the codeword 
 and  $\rho$ is a constant. We propose 
\rhoAMD  codes that provide protection in this new setting. We
 define  weak and strong \rhoAMD codes that provide  security for a random
  and an arbitrary message, respectively.
We derive concrete and asymptotic bounds for the efficiency of these codes featuring a rate upper bound of $1-\rho$ for the strong codes.
We also define the class of \LVAMD codes  that provide protection 
when 
leakage
 is   in the form of  a 
 number of  codeword components, 
and
give constructions featuring a family of strong \LVAMD codes that  asymptotically achieve the rate $1-\rho$.
 We  describe applications of \rhoAMD codes to, (i)    robust ramp secret sharing scheme and, (ii)  wiretap II channel when the adversary can
 eavesdrop a $\rho$ fraction of codeword components and 
 tamper with all components of the codeword.

\end{abstract}

\section{Introduction}

Algebraic Manipulation Detection (AMD) Codes \cite{AMD}  protect messages against additive adversarial tampering, assuming the codeword cannot be ``seen" by the adversary.
In AMD codes, 
a message  is encoded  to a codeword that is an element of a
publicly known group $\cal G$. The codeword  is stored
in a private storage which is perfectly opaque  to the adversary. 
The adversary however can 
add an {\em arbitrary }
element of $\cal G$  to the storage 
to make the decoder output a different message.  
A $\AMDdelta$-secure AMD code guarantees that any such manipulation   succeeds with probability at most $\AMDdelta$. 
Security of AMD codes has been defined for  ``weak" and ``strong" codes: weak codes provide security assuming message distribution
is uniform, while 
strong codes guarantee security 
for any message distribution. 
Weak AMD codes are primarily deterministic codes  and security relies on the randomness of the message space.  
Strong AMD codes are randomized codes and provide security for any message.
\remove{
\textcolor{red}{In  {\em systematic AMD codes} a tag  function $f$ is used to calculate a tag.
 The codeword of a  systematic strong AMD code is of the form $(\mathbf{m},\mathbf{r}, f(\mathbf{m},\mathbf{r}))$, where $\mathbf{r}$ is the explicit randomness of encoding and $f$ is the  tag.} 
 }
AMD codes have wide applications as a building block of cryptographic primitives such as robust information dispersal \cite{AMD}, 
and anonymous message transmission \cite{AMD}, and have been 
used to provide a generic construction  for
 robust secret sharing schemes from
linear secret sharing schemes \cite{AMD}.

AMD codes with leakage were first considered in \cite{LLR-AMD} where  the leakage was defined for specific parts of the encoding process.
An {\em $\alpha$-weak  AMD code with linear leakage},  also called $\alpha$-weak  LLR-AMD code, is a deterministic code that guarantees security when part of the message is leaked but  the min-entropy of the message space is at least $1-\alpha$ fraction of the message length (in bits). 
An {\em $\alpha$-strong LLR-AMD }  is a randomized code that guarantees security 
 when the randomness of encoding, although partially leaked,  has at least 
min-entropy $(1-\alpha)\log |\mathcal{R}|$ where $\mathcal{R}$ is the
 randomness set of encoding.

 \remove{
An AMD encoder takes a message, and possibly a random string (in the case of strong code), and outputs a codeword $\mathbf{x}$. The above models of leakage effectively assumes the input to the encoder is partially leaked and imposes
restrictions on the amount of remaining entropy of the input to the encoder. Our proposed leakage model however focuses on
}
In this paper we consider leakage  from the storage that holds the codeword. 
This effectively  relaxes the original model of AMD codes that required the codeword to be perfectly private.
As we will show this model turns out to be more challenging compared to LLR-AMD models 
where the leakage is in a more restricted part in the encoding process.
%
 
A more detailed relation between our model and LLR-AMD models 
 is given in Section \ref{sec: LV-AMD}.

\subsection*{Our work}
We define $\rho$-Algebraic Manipulation Detection (\rhoAMD ) codes as an extension of AMD codes 
  when the storage that holds the codeword (an element of $\mathcal{G}$), leaks up to $\rho\log|\mathcal{G}|$ bits of information about the codeword. 
  We assume 
   the adversary can apply an arbitrary function to the storage 
   and receive up to $\rho\log|\mathcal{G}|$ bits of information about the codeword.
Similar to the  original AMD codes, we define  weak   and  strong \rhoAMD codes as deterministic  and randomized codes that
 guarantee security for a uniformly distributed message and any message, respectively.
  \remove{
     and provide protection against an arbitrary message chosen by the adversary. The leaked information can be used by the adversary to choose the best offset vector and maximize their chance of success in resulting an undetectable tampering.  Weak \rhoAMD codes are deterministic codes and their security guarantee is for a uniformly chosen message.  We define strong and weak codes when the storage leaks information,  and 
     }
     
Efficiency of \rhoAMD codes is defined 
concretely (similar to \cite{AMD}) and asymptotically (using 
the {\em  rate of the code family}, which is the asymptotic 
 ratio of  the message length to the codeword length, as the message length approaches infinity ).
We prove concrete bounds for both strong and weak \rhoAMD codes and
a non-trivial upper bound $1-\rho$ on the rate of the strong \rhoAMD codes.   Comparison of bounds for different models of AMD codes is summarized in Table \ref{tb: bounds}.
     
\remove{     
Efficiency of \rhoAMD codes is defined 
concretely and asymptotically using 
{\em effective tag size} 
and the {\em  rate of the code family }, respectively. 
Tag size is  the difference between the bit length  of a codeword, and the  bit length  of a message. 
The rate of a code family is the asymptotic 
 ratio of  the message length to the codeword length, as the message length approaches infinity.
We prove lower bounds on 
the effective tag size, and  an upper bound  on the rate of the \rhoAMD codes (strong). 
}

  \begin{table}
\begin{center}
  \begin{tabular}{@{} |c|c|c|@{}}
    \hline
    codes & concrete bound & rate bound  \\ 
    \hline
    strong AMD & $G\geq\frac{M-1}{\delta^2}+1$  & $1$ \\ 
     {\bf strong \rhoAMD} & $\mathbf{G^{1-\rho}\geq\frac{M-1}{\delta^2}+1}$ & $\mathbf{1-\rho}$\\ 
     $\alpha$-strong LLR-AMD & $G\geq\frac{(M-1)(1-e^{-1})}{\delta^{\frac{2}{1-\alpha}}}+1$  & $1$\\ 
    \hline
    weak AMD & $G\geq\frac{M-1}{\delta}+1$  & $1$\\ 
    {\bf weak \rhoAMD} &$\mathbf{G\geq\frac{M-1}{\delta}+1}$ {\bf and} $\mathbf{M\geq\frac{G^\rho}{\delta}}$  & $\mathbf{1}$\\ 
    
    $\alpha$-weak LLR-AMD & $G\geq\frac{(M-1)(1-e^{-1})}{\delta^{\frac{1}{1-\alpha}}}+1$ and $G\geq\frac{M^\alpha (M-1)(1-e^{-1})}{\delta}+1$  & $\frac{1}{1+\alpha}$ \\ 
    \hline
  \end{tabular}
\caption{\label{tb: bounds}$G$ denotes the size of the group $\mathcal{G}$ that codewords live in and $M$ denotes the size of the message set $\mathcal{M}$. $\delta$ is the security parameter.}
\end{center}
\end{table}

\remove{
  \begin{table}
\begin{center}
  \begin{tabular}{@{} |c|c|c|c| @{}}
    \hline
    codes & concrete bound & rate bound & construction \\ 
    \hline
    strong AMD & $G\geq\frac{M-1}{\delta^2}+1$  & $1$ &\\ 
     strong \rhoAMD & $G^{1-\rho}\geq\frac{M-1}{\delta^2}+1$ & $1-\rho$&Theorem 2 in [2]\\ 
     $\alpha$-strong LLR-AMD & $G\geq\frac{(M-1)(1-e^{-1})}{\delta^{\frac{2}{1-\alpha}}}+1$  & $1$ &\\ 
    \hline
    weak AMD & $G\geq\frac{M-1}{\delta}+1$  & $1$ &\\ 
    weak \rhoAMD & $G\geq\frac{M-1}{\delta}+1$ and $M\geq\frac{G^\rho}{\delta}$  & $1$ &Theorem 2 in [1]\\ 
    
    $\alpha$-weak LLR-AMD & $G\geq\frac{(M-1)(1-e^{-1})}{\delta^{\frac{1}{1-\alpha}}}+1$ and $G\geq\frac{M^\alpha (M-1)(1-e^{-1})}{\delta}+1$  & $\frac{1}{1+\alpha}$ &\\ 
    \hline
  \end{tabular}
\caption{\label{tb: bounds}$G$ denotes the size of the group that codewords live in and $M$ is the size of the message set}
\end{center}
\end{table}
}

For construction, we use the relationship between \rhoAMD codes and LLR-AMD codes, to construct  (non-optimal) \rhoAMD codes,
and leave  the construction of rate optimal \rhoAMD codes as an interesting open problem.
We however define a special type of leakage in which leakage is specified by the number of codeword components that the 
adversary can select for eavesdropping.  
The model is called 
{\em limited-view \rhoAMD (\LVAMD)}.
The \LVAMD adversary is allowed to select a fraction $\rho$ of the {\em codeword components}, and 
  select their tampering (offset) vector 
 after seeing the values of the chosen components.
   This definition of limited-view adversary  was first used in \cite{ISIT-LV} where the writing power of the adversary was also parametrized. 
   We give an explicit construction  of strong \LVAMD codes
that achieve rate $1-\rho$, using 
an AMD code and a wiretap II code as building blocks. We note that this rate is achievable for large constant size alphabets, if we allow a seeded encoder involving a universal hash family (see \cite{eAWTP}). That is the alphabet size depends on the closeness to the actual capacity value.
Also we do not know if $1-\rho$ is the capacity of strong \LVAMD codes.
Finding the capacity of strong \LVAMD codes however is an open question as the type of leakage (component wise) is more restricted than strong \rhoAMD codes.
 \remove{
 REMOVE: We define the \textcolor{blue}{limitted-view} adversary by a set of functions, each \textcolor{red}{mapping  their view of the codeword  to an offset vector  to be added to the codeword} \textcolor{blue}{detailing an adversarial strategy: choose a subset of the codeword components to read and depending on the value of the components select an offset vector  to be added to the codeword}. 
 This approach was used in  non-malleable code  \cite{DzPiWi}.
 } 
%
We also construct a family of weak \LVAMD codes that achieve rate $1$ for any leakage parameter $\rho$.

We consider two applications. 
The first application can be seen as parallel to the application of the original AMD codes to robust secret sharing scheme.  The second application is a new variation of active adversary wiretap channel II.

\noindent {\bf Robust ramp secret sharing scheme.}
\remove{
It is shown in \cite{AMD} that a $\delta$-robust secret sharing scheme can be constructed from a linear SSS by adding a layer of AMD coding before the secret is divided into shares. The privacy of the SSS guarantees the ``abstract private storage'' and the linearity turns the corruption on shares into an additive translate on the AMD codeword.
}
A $(t,r,N)$-ramp secret sharing scheme \cite{ramp SSS,new ramp}  
is a secret sharing scheme with two thresholds, $t$ and $r$, such that any $t$ or less shares do not leak any information about the secret while any $r$ or more shares 
reconstruct the secret
and if the number $a$ of shares  is in  between $t$ and $r$,  an $\frac{a-t}{r-t}$ fraction of information  of the secret will be leaked.
We define a \textit{robust ramp secret sharing scheme} as a ramp secret sharing scheme with an additional ($\rho,\delta$)-robustness property which requires that  the  probability of reconstructing a  wrong secret, if up to $t+\lfloor\rho(r-t)\rfloor$ shares are controlled by an active adversary, is bounded by $\delta$. Here  $\rho$ is a constant.
We will show that a $(t,r,N,\rho,\delta)$-robust secret sharing scheme can be constructed from a linear $(t,r,N)$-ramp secret sharing scheme, by first encoding
the message using a \rhoAMD code with security parameter $\delta$, and then using the linear ramp secret sharing scheme to generate shares.

\noindent {\bf 
Wiretap II with an algebraic  manipulation adversary.} 
Wiretap model of communication was proposed by Wyner \cite{WT}.  In wiretap II setting \cite{WtII}, 
the goal is to provide secrecy against a passive adversary who can adaptively select a fraction $\rho$ of transmitted codeword components  to eavesdrop.
We consider active wiretap II adversaries that in addition to eavesdropping the channel, algebraically manipulate the communication by adding a noise (offset) vector to the sent codeword.  The code must protect against eavesdropping and also detect tampering. 
An algebraic manipulation wiretap II code is a wiretap II code with security against an eavesdropping adversary and so the rate upper bound for wiretap II codes is applicable. Our construction of \LVAMD codes gives a family
of algebraic manipulation wiretap II codes which achieve this rate upper bound and so the construction is capacity-achieving. The result effectively shows that algebraic manipulation  detection in this case can be achieved for ``free'' (without rate loss), asymptotically

Table \ref{tb: constructions and applications} summarizes the code constructions and applications. 

  \begin{table}
\begin{center}
  \begin{tabular}{@{} |c|c|c|@{}}
    \hline
    codes constructed & asymptotic rate & applications  \\ 
    \hline
     {\bf strong \rhoAMD} & {\bf N.A. }& {\bf $(\rho,\delta)$-robust ramp secret sharing}\\ 
     {\bf strong \LVAMD} & $\mathbf{1-\rho}$  & {\bf $(\rho,0,\delta)$-algebraic adversary wiretap II}\\ 
    \hline
     {\bf weak \rhoAMD} & {\bf N.A. } &  {\bf N.A. }\\ 
    
    {\bf weak \LVAMD} & $\mathbf{1}$  &  {\bf N.A. } \\ 
    \hline
  \end{tabular}
\caption{\label{tb: constructions and applications}Summary of codes constructed in this paper and their applications.}
\end{center}
\end{table}

 
\subsection*{Related works}
\vspace{1.5mm}
\noindent {\bf Related works.}
AMD codes were proposed in \cite{AMD} and have found numerous applications. 
 A work  directly comparable to ours  is \cite{LLR-AMD} where 
LLR-AMD code with different leakage models for weak and strong codes are introduced. 
Our leakage model uses a single leakage model for both weak and strong codes and is a natural generalization of the original AMD codes.
The relation between our model and 
 LLR-AMD codes
 is given in Section \ref{sec: LV-AMD}. 
More generally, there is a large body of work on modelling leakage and designing leakage resilient systems.
A survey 
 can be found in \cite{LRcrypto}.     
 
Ramp secret sharing  schemes (ramp SSS) are introduced in \cite{ramp SSS}.  Robust secret sharing schemes (robust SSS) are well studied (see for example \cite{AMD}). To our knowledge robust ramp secret sharing schemes (robust ramp SSS) have not been considered before. In a robust SSS, robustness is defined only when the number of the compromised players  is
below the privacy threshold of the underling SSS. Our definition of robust ramp SSS has robustness guarantee even when the number of compromised players is
 bigger than the privacy threshold. 
 
Wiretap II model with active adversary was first studied in \cite{Lai Lifeng}, where the eavesdropped components and tampered components are restricted to be the same set.
A general model of wiretap II adversaries with additive manipulation was defined in \cite{AWTP}. In this model (called adversarial wiretap or AWTP) the adversary can
read a fraction $\rho_r$, and add noise to a fraction $\rho_w$,  of the codeword components. The goal  of the encoding scheme is to provide secrecy  and
guarantee reliability (message recovery) against this adversary. A variation of AWTP called eAWTP is studied in \cite{eAWTP}, where erasure of codeword components instead of additive tampering is considered. Interestingly, both AWTP and eAWTP have the same capacity $1-\rho_r-\rho_w$.  The alphabet of known capacity-achieving codes are, $\mathcal{O}({\frac{1}{\xi^4}}^{\frac{1}{\xi^2}})$ for AWTP codes and $\mathcal{O}(2^{\frac{1}{\xi^2}})$ for eAWTP codes, respectively, where $\xi$ is the difference of the actual rate and capacity \cite{eAWTP}. 
The adversary of algebraic manipulation wiretap II codes defined in this paper can  be seen as 
 the AWTP adversary with $\rho_r =\rho$ and $\rho_w=1$, yielding $1-\rho_r-\rho_w<0$.  
In this case recovering the message is impossible. Our results on algebraic manipulation wiretap II show that a weaker goal against active attack, that is {\em to detect } manipulation of the message, is achievable and can be achieved with capacity $1-\rho$, which is the same as the capacity of wiretap II codes with no security against active attacks.



\bigskip


\noindent
{\em Organization:} In Section \ref{sec: intro}, we give notations and  introduce AMD codes (with\slash without leakage) and wiretap II codes. In Section \ref{sec: rho-AMD}, we define \rhoAMD codes and derive efficiency bounds. In Section \ref{sec: LV-AMD construction}, we study
\LVAMD codes and give concrete constructions. In Section \ref{sec: applications}, we give two applications. 


\section{Preliminaries} \label{sec: intro}
Calligraphy letters $\mathcal{X}$ denote sets and their corresponding capital letters denote the cardinality, $|\mathcal{X}|=X$. Boldface letters $\mathbf{x}$ denote vectors. $\mathbf{x}_{|S}$ denotes the sub-vector of $\mathbf{x}$ consisting of the components specified by the index set $S$. $[n]$ denotes $\{1,2,\cdots,n\}$. Capital boldface letters $\mathbf{X}$ denote 
random variables,  and $\mathbf{X}\leftarrow\mathcal{X}$ denotes sampling of the  random variable $\mathbf{X}$ 
 from the set $\mathcal{X}$,  with $\mathbf{X}\stackrel{\$}{\leftarrow}\mathcal{X}$ denoting a uniform distribution in sampling.
The statistical distance between $\mathbf{X}$ and $\mathbf{Y}$ that are both defined over the set $\mathcal{W}$, is defined as,
$$
\SD(\mathbf{X},\mathbf{Y})\triangleq \dfrac{1}{2}\sum_{\mathbf{w} \in \mathcal{W}}|\bP[\mathbf{X}=\mathbf{w}]-\bP[\mathbf{Y}=\mathbf{w}]|.
$$
We say $\mathbf{X}$ and $\mathbf{Y}$ are $\delta$-close if $\SD(\mathbf{X},\mathbf{Y})\leq \delta$. 
The \textit{min-entropy} $\sfH_\infty(\mathbf{X})$ of a random variable $\mathbf{X}\leftarrow\mathcal{X}$ is 
$$
\sfH_\infty(\mathbf{X})=-\log\max_{\mathbf{x} \in \mathcal{X}}\bP[\mathbf{X}=\mathbf{x}].
$$
The \textit{(average) conditional min-entropy} $\tilde{\sfH}_\infty(\mathbf{X}|\mathbf{Z})$  of $\mathbf{X}$ conditioned on $\mathbf{Z}$  is defined \cite{Dodis Fuzzy} as,
$$
\tilde{\sfH}_\infty(\mathbf{X}|\mathbf{Z})=-\log\left( \mathbb{E}_{\mathbf{Z}=\mathbf{z}}\max_\bfx \bP[\mathbf{X} = \bfx|\mathbf{Z}=\mathbf{z}]\right).
$$
The following bound on the amount of information about one variable that can leak through a correlated variable is proved in \cite{Dodis Fuzzy}.
\begin{lemma}\label{lem: conditional min-entropy}\cite{Dodis Fuzzy}Let $\mathbf{X}\leftarrow\mathcal{X}$ and $\mathbf{Z}\leftarrow\mathcal{Z}$ with $\ell=\log |\mathcal{Z}|$. Then
$$
\tilde{\sfH}_\infty(\mathbf{X}|\mathbf{Z})\geq \sfH_\infty(\mathbf{X})-\ell.
$$
\end{lemma}
\begin{definition}\label{def: AMD}
An $(M,G,\AMDdelta)$-algebraic manipulation detection code, or $(M,G,\delta)$-AMD code for short, is a probabilistic encoding map $\mbox{Enc}: \mathcal{M}\rightarrow \mathcal{G}$ from a set $\mathcal{M}$ of size $M$ to an (additive) group $\mathcal{G}$ of order $G$, together with a deterministic decoding function $\mbox{Dec}: \mathcal{G}\rightarrow \mathcal{M}\bigcup\{\perp\}$ such that $\mbox{Dec}(\mbox{Enc}(\mathbf{m}))=\mathbf{m}$ with probability $1$ for any $\mathbf{m}\in\mathcal{M}$.
 The security of an AMD code requires that for any $\mathbf{m}\in\mathcal{M}$, $\Delta\in\mathcal{G}$, $\bP[\mbox{Dec}(\mbox{Enc}(\mathbf{m})+\Delta)\notin\{\mathbf{m},\perp\}]\leq\AMDdelta$.

\remove{An AMD code is called {\em systematic } if $\mathcal{M}$ is a group, and the encoding is of the form
$$
\mbox{Enc}: \mathcal{M}\rightarrow \mathcal{M}\times \mathcal{G}_1\times\mathcal{G}_2, \mathbf{m}\mapsto(\mathbf{m},\mathbf{r},f(\mathbf{r},\mathbf{m}))
$$
for some (tag) function $f$ and $\mathbf{r}\stackrel{\$}{\leftarrow}\mathcal{G}_1$. The decoding function of a systematic AMD code is naturally given by $\mbox{Dec}(\mathbf{m}^{'},\mathbf{r}^{'},\sigma^{'})=\mathbf{m}^{'}$ if $\sigma^{'}=f(\mathbf{r}^{'},\mathbf{m}^{'})$ and $\perp$ otherwise.}
\end{definition}

The AMD code  above is said to  provide 
{\em strong security}. 
{\em Weak AMD} codes provide security 
for randomly chosen messages. 
Efficiency of  $(M,G,\AMDdelta)$-AMD codes   is measured by the \textit{effective tag size}
which 
is defined as the minimum tag length $\min\{\log_2 G\}-u$, where the minimum is over all $(M,G,\delta)$-AMD codes with $M\geq2^u$. 
Concrete lengths are important in practice, and additionally, 
  the asymptotic  rate 
  (defined as the limit of the ratio of message length to codeword length as the length grows to infinity) of   both weak and strong AMD codes 
  has been  shown \cite{AMD}  to be 
   $1$.

\begin{lemma}\cite{AMD}\label{lem: AMD bounds} Any weak, respectively strong, $(M,G,\delta)$-AMD code satisfies
$$
G\geq\frac{M-1}{\AMDdelta}+1, \mbox{ respectively, } G\geq\frac{M-1}{\AMDdelta^2}+1.
$$
\end{lemma}
The following construction is optimal with respect to effective tag size. 

\begin{construction}\label{ex: AMD} \cite{AMD}: Let $\mathbb{F}_q$ be a field of size $q$ and characteristic $p$, and let $d$ be any integer such that $d+2$ is not divisible by $p$. Define the encoding function, 
$$
\mbox{Enc}: \mathbb{F}_q^d\rightarrow \mathbb{F}_q^d\times \mathbb{F}_q\times\mathbb{F}_q, \mathbf{m}\mapsto(\mathbf{m},\mathbf{r},f(\mathbf{r},\mathbf{m})),\mbox{ where } f(\mathbf{r},\mathbf{m})=\mathbf{r}^{d+2}+\sum_{i=1}^dm_i\mathbf{r}^i.
$$
The decoder Dec verifies a tagged message $(\mathbf{m},\mathbf{r},t )$ by comparing $t= f(\mathbf{r},\mathbf{m})$
and outputs $\mathbf{m}$ if agree; $\perp$ otherwise. (Enc,Dec) gives a $(q^d,q^{d+2},\frac{d+1}{q})$-AMD code. 
\end{construction}

\begin{definition}[strong LLR-AMD]\cite{LLR-AMD}\label{def: strong LLR-AMD}
A  randomized code with encoding function $\mbox{Enc}:\mathcal{M}\times\mathcal{R}\rightarrow \mathcal{X}$ and decoding function $\mbox{Dec}:\mathcal{X}\rightarrow \mathcal{M}\bigcup\{\perp\}$ is a $(M,X,|\mathcal{R}|,\alpha,\delta)$-strong LLR-AMD code if for any $\mathbf{m}\in\mathcal{M}$ and any $\mathbf{r}\in\mathcal{R}$, $\mbox{Dec}(\mbox{Enc}(\mathbf{m},\mathbf{r}))=\mathbf{m}$, and for any adversary $\mathbb{A}$ and variables $\mathbf{R}\stackrel{\$}{\leftarrow}\mathcal{R}$ and $\mathbf{Z}$ such that $\tilde{\sfH}_\infty(\mathbf{R}|\mathbf{Z})\geq (1-\alpha)\log |\mathcal{R}|$, it holds for any $\mathbf{m}\in\mathcal{M}$:
\begin{equation}\label{eq: strong LLR-AMD}
\bP
[\mbox{Dec}(\mbox{Enc}(\mathbf{m},\mathbf{R})+\mathbb{A}(\mathbf{Z}))\notin\{\mathbf{m},\perp\}]\leq\delta,
\end{equation}
where the probability is over the randomness of encoding.
\end{definition}

\begin{definition}[weak LLR-AMD]\cite{LLR-AMD}\label{def: weak LLR-AMD}
A deterministic code with encoding function $\mbox{Enc}:\mathcal{M}\rightarrow \mathcal{X}$ and decoding function $\mbox{Dec}:\mathcal{X}\rightarrow \mathcal{M}\bigcup\{\perp\}$ is a $(M,X,\alpha,\delta)$-weak LLR-AMD code 
if for any $\mathbf{m}\in\mathcal{M}$, $\mbox{Dec}(\mbox{Enc}(\mathbf{m}))=\mathbf{m}$, 
and for any adversary $\mathbb{A}$ and variables $\mathbf{M}\leftarrow\mathcal{M}$ and $\mathbf{Z}$ such that $\tilde{H}_\infty(\mathbf{M}|\mathbf{Z})\geq (1-\alpha)\log |\mathcal{M}|$, it holds:
\begin{equation}\label{eq: weak LLR-AMD}
\bP
[\mbox{Dec}(\mbox{Enc}(\mathbf{M})+\mathbb{A}(\mathbf{Z}))\notin\{\mathbf{M},\perp\}]\leq\delta,
\end{equation}
where the probability is over the randomness of the message.
\end{definition}
In the above two definitions, leakages are  from randomness (bounded by  $\tilde{\sfH}_\infty(\mathbf{R}|\mathbf{Z})$ $\geq (1-\alpha)\log |\mathcal{R}|$) and message space (bounded by $\tilde{\sfH}_\infty(\mathbf{M}|\mathbf{Z})\geq (1-\alpha)\log |\mathcal{M}|$), respectively.

\subsubsection*{Wiretap II codes.}
Wiretap II model \cite{WtII} of secure communication 
considers a scenario where Alice wants to send messages to Bob over a reliable channel
that is eavesdropped by an adversary,  Eve.
The adversary can read a fraction $\rho$ of the transmitted codeword components, 
 and is allowed to choose 
 any  subset  (the right size)   of their choice.
 A wiretap II code provides information-theoretic secrecy for message transmission against this adversary. 

\begin{definition}\label{def: WtII} 
A $(\rho,\varepsilon)$ wiretap II code, or $(\rho,\varepsilon)$-WtII code for short, is a probabilistic encoding function $\mbox{Enc}: \mathbb{F}_q^k\rightarrow \mathbb{F}_q^n$, together with a deterministic decoding function $\mbox{Dec}: \mathbb{F}_q^n\rightarrow \mathbb{F}_q^k$ such that $\mbox{Dec}(\mbox{Enc}(\mathbf{m}))=\mathbf{m}$ for any $\mathbf{m}\in\mathbb{F}_q^k$. The security of a $(\rho,\varepsilon)$-WtII code requires that for any $\mathbf{m}_0,\mathbf{m}_1\in\mathbb{F}_q^k$, any $S\subset [n]$ of size $|S|\leq n\rho$,
\begin{equation}\label{eq: WtII security}
\SD(\mbox{Enc}(\mathbf{m}_0)_{|S};\mbox{Enc}(\mathbf{m}_1)_{|S})\leq \varepsilon
\end{equation}
A rate $R$ is achievable if there exists a family of $(\rho,\varepsilon)$-WtII codes with encoding and decoding functions $\{\mbox{Enc}_n,\mbox{Dec}_n\}$ such that $\lim_{n\rightarrow\infty}\frac{k}{n}=R$.
\end{definition}
The above definition of security is in line with \cite{AWTP} and is stronger than the original definition  \cite{WtII}, and  also the definition in \cite{invertible extractors}.

\remove{
 \textcolor{red}{The following rate upper bound was proved for wiretap II codes with a weaker security,
 and hence will also be 
  an upper bound for WtII codes with our stronger security definition. The weaker security assumes a uniformly distributed message $\mathbf{M}$ and is defined as $\lim_{k\rightarrow\infty}\frac{H_\infty(\mathbf{M}|\mathbf{Z})}{k}=1$, where $\mathbf{Z}$ is the view of the adversary.}
}

\begin{lemma} \cite{AWTP}\label{lem: WtII upper bound} The achievable rate of $(\rho,0)$-WtII codes is upper bounded by $1-\rho$.
\end{lemma}

When $\varepsilon=0$ is achieved in 
(\ref{eq: WtII security}), the distribution of any $\rho$ fraction of the codeword components is independent of the message. This is achieved, for example, by the following construction of wiretap II codes.

\begin{construction}\label{ex: WtII}\cite{WtII} Let $G_{(n-k)\times n}$ be a generator matrix of a $[n,n-k]$ MDS code over $\mathbb{F}_q$. Append $k$ rows to $G$ such that the obtained matrix $\left [\begin{array}{c} G\\\tilde{G}\end{array} \right ]$ is non-singular. Define the encoder WtIIenc as follows.
$$
\mbox{WtIIenc}(\mathbf{m})=[\mathbf{r},\mathbf{m}]\left [\begin{array}{c} G\\\tilde{G}\end{array} \right ],\mbox{ where }\mathbf{r}\stackrel{\$}{\leftarrow}\mathbb{F}_q^{n-k}.
$$
WtIIdec uses a parity-check matrix $H_{k\times n}$ of the MDS code to  first compute the syndrome, $H\mathbf{x}^T$, and then map the syndrome back to the message using the one-to-one correspondence between syndromes and messages. The above construction gives a family of   $(\rho,0)$-WtII codes for  $\rho=\frac{n-k}{n}$.
\end{construction}



\section{AMD codes for leaky storage}\label{sec: rho-AMD}
We consider 
codes over a finite field $\mathbb{F}_q$, where $q$ is a prime power, and  assume 
message set $\mathcal{M}=\mathbb{F}_q^k$ and the storage stores an element of the group 
$\mathcal{G}=\mathbb{F}_q^n$.

\subsection{Definition of \rhoAMD}\label{sec: LV-AMD}
\begin{definition}
An $(n,k)$-coding scheme consists of two functions: a randomized  
{\em encoding function } $\mbox{Enc}:\mathbb{F}_q^k\rightarrow\mathbb{F}_q^n$, 
and deterministic {\em decoding function} $\mbox{Dec}:\mathbb{F}_q^n\rightarrow\mathbb{F}_q^k\cup\{\perp\}$, 
satisfying 
$\bP[\mbox{Dec}(\mbox{Enc}(\mathbf{m}))=\mathbf{m}]=1$, for any $\mathbf{m}\in\mathbb{F}_q^k$.
Here  probability is taken over the randomness of the encoding algorithm. 

The {\em information rate } of an $(n,k)$-coding scheme is $\frac{k}{n}$.
\end{definition}

We now define our leakage model and  codes that detect manipulation in  presence of this leakage.
{ Let   $\mathbf{X}=\mbox{Enc}(\mathbf{m})$ for a message $\mathbf{m}\in {\cal M}$,  and $\mathbb{A}_\mathbf{Z}$ denote an adversary with access to a variable $\mathbf{Z}$, representing the leakage of information about the codeword.

\begin{definition}[\rhoAMD]\label{def: AMD with leakage} An $(n,k)$-coding scheme is called a {\em  strong \rhoAMD code} with security parameter $\delta$ if $\bP[\mbox{Dec}(\mathbb{A}_\mathbf{Z}(\mbox{Enc}(\mathbf{m})))\notin\{\mathbf{m},\perp\}]\leq\delta$ for any message $\mathbf{m}\in\mathbb{F}_q^k$
and  adversary $\mathbb{A}_\mathbf{Z}$ 
whose leakage variable $\mathbf{Z}$  satisfies 
 $\tilde{\sfH}_\infty(\mathbf{X}|\mathbf{Z})\geq \sfH_\infty(\mathbf{X})-\rho n\log q$,
  and  
is allowed to choose any 
offset vector in  $\mathbb{F}_q^n$ to    add to the codeword.

\noindent
The code  is called a {\em weak \rhoAMD code} if 
 security holds for  $\mathbf{M}\stackrel{\$}{\leftarrow}\mathbb{F}_q^k$ (rather than an arbitrary message distribution).
The encoder in this case is  
 deterministic and the probability of outputing a different message is over the randomness of the message.

\noindent
A {\em  family $\{(\mbox{Enc}_n,\mbox{Dec}_n)\}$ of \rhoAMD codes is a set of $(n,k(n))$-coding schemes} indexed by the codeword length $n$, where for any value of $\delta$, there is an $N\in \mathbb{N}$ such that for all $n\geq N$, $(\mbox{Enc}_n,\mbox{Dec}_n)$ is a  \rhoAMD code with security parameter $\delta$.

\noindent A {\em rate $R$ is achievable} if there exists a family $\{(\mbox{Enc}_n,\mbox{Dec}_n)\}$ of \rhoAMD codes such that $\lim_{n\rightarrow\infty}\frac{k(n)}{n}=R$ as $\delta$ approaches $0$
\end{definition}

Our definition bounds the  amount of leakage in comparison with  an adversary who observes up to $\rho n$ components of the
stored codeword.  We call this latter adversary  a  {\em Limited-View (LV) adversary} \cite{ISIT-LV}. 
According to Lemma \ref{lem: conditional min-entropy},  the min-entropy of the stored codeword given  an LV-adversary 
will be   $\tilde{\sfH}_\infty(\mathbf{X}|\mathbf{Z})\geq \sfH_\infty(\mathbf{X})-\rho n\log q$.
We require the same min-entropy be left in the codeword, 
for an arbitrary leakage  variable $\bfZ$ accessible to the adversary. 
}

\remove{

Our definition  captures the intuition that the information leakage, from the storage $\mathcal{G}$ to the adversary $\mathbb{A}$,  is bounded by  
$\rho\log|\mathcal{G}|$. 
 Consider the  setting where the codeword is stored as a vector in $\mathbb{F}_q^n$,
  and the adversary $\mathbb{A}$ can only choose  $\rho n$ components of the vector for 
  eavesdropping.  Lemma \ref{lem: conditional min-entropy} asserts that such adversary satisfies $\tilde{H}_\infty(\mathbf{X}|\mathbf{Z})\geq H_\infty(\mathbf{X})-\rho n\log q$, where $\mathbf{Z}$ is any $\rho$ fraction of codeword components. An adversary with this type of leakage has been called a  {\em Limited-View (LV) adversary} \cite{ISIT-LV}. 

  In Section \ref{sec: LV-AMD construction}, we give constructions of \rhoAMD codes for an LV adversary.
}

\begin{figure}
\centerline{\includegraphics[scale=0.45]{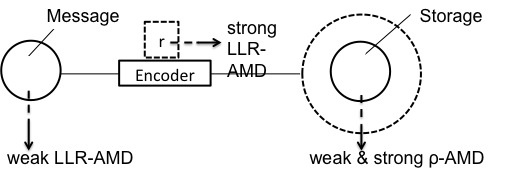}}
\caption{\label{fig: comparison}The arrow shows the part of the system that leaks.}
\end{figure}

 Fig. \ref{fig: comparison}  shows places of  leakage 
 in  AMD encoding  in  
our model, and  the models in Definition \ref{def: strong LLR-AMD} and Definition \ref{def: weak LLR-AMD}. 

\begin{proposition}\label{prop1}
{Let 
$\mathbf{X}$ denote a random variable representing the codeword of a message $\mathbf{m}$ ($\mathbf{M}$ for weak codes), and $\mathbf{Z}$ denote the leakage variable of the 
adversary $\mathbb{A}_\mathbf{Z}$ who uses the leakage information to construct the best offset vector to make the decoder output a different message. 
For a $\rho$-AMD code with security parameter $\delta$, we have $\tilde{H}(\mathbf{X}|\mathbf{Z})\geq\log\frac{1}{ \delta}$. }
\end{proposition}

\begin{proof}
{
We write the proof 
for strong $\rho$-AMD codes. (The proof for weak \rhoAMD codes follows similarly.) According to the security definition of $\rho$-AMD codes, we have
$$
\mbox{Pr}[\mbox{Dec}(\mathbb{A}_\mathbf{Z}(\mathbf{X}))\notin\{\mathbf{m},\perp\}]\leq\delta,
$$
where the probability is over the randomness of $\mathbf{X}$, and is the  expectation over $\mathbf{z}\in\mathcal{Z}$. 
If the adversary with the leakage variable 
$\mathbf{Z}=\mathbf{z}$ can correctly guess the value $\mathbf{x}$ of $\mathbf{X}$, then a codeword $\mathbf{x}'$ corresponding to another message $\mathbf{m}'$ can be constructed  to cause the decoder to output $\mathbf{m}'$, by using $\mathbb{A}_\mathbf{z}(\mathbf{X})=\mathbf{X}+(\mathbf{x}'-\mathbf{x})$. We then have
$$
\mbox{Pr}[\mbox{Dec}(\mathbb{A}_\mathbf{Z}(\mathbf{X}))\notin\{\mathbf{m},\perp\}|\mathbf{Z}=\mathbf{z}]\geq\max_{\mathbf{x}}\mbox{Pr}[\mathbf{X}=\mathbf{x}|\mathbf{Z}=\mathbf{z}],
$$
which by taking expectation over $\mathbf{z}\in\mathcal{Z}$ yields 
$$
\mathbb{E}_\mathbf{z}\left(\mbox{Pr}[\mbox{Dec}(\mathbb{A}_\mathbf{Z}(\mathbf{X}))\notin\{\mathbf{m},\perp\}|\mathbf{Z}=\mathbf{z}]\right)\geq\mathbb{E}_\mathbf{z}\left(\max_{\mathbf{x}}\mbox{Pr}[\mathbf{X}=\mathbf{x}|\mathbf{Z}=\mathbf{z}]\right)=2^{-\tilde{H}(\mathbf{X}|\mathbf{Z})},
$$
The last equality follows from the definition of conditional min-entropy. The desired inequality then follows directly from the security definition of $\rho$-AMD codes as follows.
$$
2^{-\tilde{H}(\mathbf{X}|\mathbf{Z})}\leq\mbox{Pr}[\mbox{Dec}(\mathbb{A}_\mathbf{Z}(\mathbf{X}))\notin\{\mathbf{m},\perp\}|\mathbf{Z}=\mathbf{z}]\leq \delta\Longleftrightarrow \tilde{H}(\mathbf{X}|\mathbf{Z})\geq\log\frac{1}{ \delta}.
$$
\qed
}

\end{proof}

\begin{definition}
Let  $\cal C$ denote the set of codewords  of a code,  and {${\cal C}_\mathbf{m}$ denote the set of codewords corresponding to the message $\mathbf{m}$, i.e. ${\cal C}_\mathbf{m} = \{ \mbox{Enc}(\mathbf{m},\mathbf{r})| \mathbf{r}\in {\cal R}\} $}.
 A  randomised encoder is  called \textit{regular}  if  $|{\cal C}_\mathbf{m}| =  |{\cal R}|$ for all $\mathbf{m}$. 
\end{definition}
 We note that  because the code has zero decoding error when there is no
 adversary corruption, we have 
\begin{equation}\label{eq_le_10}
 {\cal C}_\mathbf{m}   \cap C_{\mathbf{m}'}  =\emptyset,\; \forall \mathbf{m},\mathbf{m}' \in {\cal M}.
\end{equation}
{This means that for  regular randomised encoders, a codeword  uniquely determines a pair $(\mathbf{m},\mathbf{r})$.} 
Assuming that the randomized encoder uses $r$ uniformly distributed bits,  
 { the random variable {$\mathbf{X}=\mbox{Enc}(\mathbf{m},\mathbf{R})$  is flat over} ${\cal C}_\mathbf{m}$. 

\begin{lemma}\label{lem: strong comparison}
The relations between Strong LLR-AMD codes and strong \rhoAMD codes are as follows.

\begin{enumerate}
\item 
If there exists a regular randomized encoder for a $(q^k,q^n,2^r,\alpha,\delta)$-strong LLR-AMD code, then there is an encoder for strong \rhoAMD code with security parameter $\delta$ and leakage parameter $\rho$ where 
{$\rho\leq\frac{\alpha r}{n\log q}$}. 
\item 
If there exists a regular randomized encoder for a strong \rhoAMD code with security parameter $\delta$ and leakage parameter $\rho$,  then there is an encoder for a $(q^k,q^n,2^r,\alpha,\delta)$-strong LLR-AMD code with $\alpha$ and $r$ where 
{$\alpha\leq\frac{n\rho\log q}{r}$ and $r\geq\log \frac{1}{\delta}+n\rho\log q$.} 
\end{enumerate}

\end{lemma}
}
Proof of Lemma \ref{lem: strong comparison} is given in Appendix \ref{sec: proof of strong comparison}.


In \cite{LLR-AMD},  it is shown that the 
 optimal AMD code in Construction \ref{ex: AMD} gives a $(q^d,q^{d+2},q,\alpha,\frac{d+1}{q^{1-\alpha}})$-strong LLR-AMD code. The parameters of this 
  LLR-AMD code are 
  $k=d$, $n=d+2$, $r=\log q$ and $\delta=\frac{d+1}{q^{1-\alpha}}$.
   A simple mathematical manipulation of these equations gives 
$\alpha=1-\log_q \frac{n-1}{\delta}$, and substituting them into 
Lemma \ref{lem: strong comparison}, item 1, we obtain 
$$
\rho\leq\frac{(1-\log_q \frac{n-1}{\delta})\log q}{n\log q }=\frac{1-\log_q \frac{n-1}{\delta}}{n }.
$$
This results in the following.
\begin{corollary}
The code in Construction \ref{ex: AMD} is a strong \rhoAMD code with $k=d$, $n=d+2$, security parameter $\delta$ and leakage parameter  
{$\rho\leq\frac{1-\log_q \frac{n-1}{\delta}}{n }$}.
\end{corollary}

It is easy to see that $\rho<\frac{1}{n}$. 
Thus the resulting construction of strong \rhoAMD codes can only tolerate a 
very small leakage. 
Moreover the upper bound on $\rho$ vanishes as $n$ goes to infinity and so this construction cannot give a 
non-trivial family of 
strong \rhoAMD code. 
We note that the same construction resulted in  a family of strong LLR-AMD codes  with asymptotic rate $1$. 

\begin{lemma}\label{lem: weak comparison} 
The relations between weak LLR-AMD codes and weak \rhoAMD codes are 
as follows.
\begin{enumerate}
\item 
A $(q^k,q^n,\alpha,\delta)$-weak LLR-AMD code is a weak \rhoAMD code with security parameter $\delta$ and leakage parameter $\rho$ satisfying $\rho\leq\frac{\alpha k}{n}$.  
\item 
A weak \rhoAMD code with security parameter $\delta$ and leakage parameter $\rho$ is a $(q^k,q^n,\alpha,\delta)$-weak LLR-AMD code satisfying $\alpha\leq\frac{\rho n}{ k}$.
\end{enumerate}
\end{lemma}
Proof of Lemma \ref{lem: weak comparison} is given in Appendix \ref{sec: proof of weak comparison}.

 A construction of $(q^d,q^{d+1},\alpha,\frac{2}{q^{1-\alpha d}})$ weak LLR-AMD codes is given in \cite[Theorem 2]{LLR-AMD}. 
 The 
 code has parameters $k=d$, $n=d+1$ and $\delta=\frac{2}{q^{1-\alpha d}}$. A simple mathematical manipulation of these equations gives $\alpha=\frac{1-\log_q \frac{2}{\delta}}{n-1}$, and so from  
Lemma \ref{lem: weak comparison}, item 1, we obtain 
$$
\rho\leq\frac{(\frac{1-\log_q \frac{2}{\delta}}{n-1})(n-1)}{n}=\frac{1-\log_q \frac{2}{\delta}}{n}.
$$
\begin{corollary}
The code 
 in \cite[Theorem 2]{LLR-AMD} is a weak \rhoAMD code with $k=d$, $n=d+1$, security parameter $\delta$ and leakage parameter 
 $\rho\leq\frac{1-\log_q \frac{2}{\delta}}{n}$.
\end{corollary} 
This construction   gives \rhoAMD codes with small $\rho$,  and cannot be used to construct 
a family of \rhoAMD codes for $\rho>0$.



\subsection{Efficiency bounds for \rhoAMD codes}\label{sec: upper bound}

\begin{theorem}\label{lem: generalised upper bound} If an $(n,k)$-coding scheme is a strong \rhoAMD code with security parameter $\delta$, then,
 
\begin{equation}
k\leq n(1-\rho)+\frac{2\log \delta-1}{\log q}.
\end{equation}
The achievable rate of strong \rhoAMD codes is upper bounded by $1-\rho$.
\end{theorem}

\begin{proof}  
Consider a strong \rhoAMD code with security parameter $\delta$.
By Proposition \ref{prop1}, $\tilde{H}_\infty(\bfX |\bfZ)\geq\log\frac{1}{\delta}$ should hold for any $\mathbf{Z}$ satisfying $\tilde{H}_\infty(\bfX |\bfZ)\geq H_\infty( \bfX) -\rho n\log q$. In particular, the inequality should hold for $\mathbf{Z}$ such that $\tilde{H}_\infty(\bfX |\bfZ)= H_\infty( \bfX) -\rho n\log q$. We then have $H_\infty( \bfX) -\rho n\log q\geq\log\frac{1}{\delta}$. On the other hand, we always have $\log |{\cal C}_\mathbf{m}|\geq H_\infty( \bfX)$, where ${\cal C}_\mathbf{m}$ denotes the set of codewords corresponding to message $\mathbf{m}$, which is the support of $\bfX$.
This gives the following lower bound on $|{\cal C}_\mathbf{m}|$. 
  

\begin{equation}\label{eq: coset bound}
{|\cal C}_\mathbf{m}| \geq \frac{2^{\rho n \log q}}{\delta} = \frac{q^{\rho n}}{\delta}.
\end{equation}
Now consider the adversary randomly choose an offset $\Delta\neq 0^n$, we have 

\begin{equation}
\begin{split}
\delta & \geq \bP[\mbox{Enc}(\bfm) + \Delta \in \cup_{\bfm' \neq \bfm} E_{\bfm'} ]\\
&\geq \frac{| \bigcup_{\bfm' \neq \bfm} E_{\bfm'} |}{|\mathbb{F}_q^n| - 1}\\
& \overset{(\ref{eq_le_10}),(\ref{eq: coset bound})}{\geq} \frac{(q^{k}-1)\cdot \frac{q^{\rho n}}{\delta}}{q^{n}-1}. \\
\end{split}
\end{equation}

Therefore,
$$
k\leq n(1-\rho)+\frac{2\log \delta-1}{\log q}.
$$
\qed
\end{proof}

\begin{proposition}\label{pro: weak bound}
If an $(n,k)$-coding scheme is a weak \rhoAMD code with security parameter $\delta$, then $q^{\rho n-k}\leq\delta$ and $\frac{q^{k}-1}{q^{n}-1}\leq \delta$.
\end{proposition}
Proof of Proposition \ref{pro: weak bound} is given in Appendix \ref{sec: proof of weak bound}.

\section{Limited-View \rhoAMD codes} 
\label{sec: LV-AMD construction}

\remove{
In this subsection, we give code constructions for LV adversary, namely, an adversary who reads a $\rho$ fraction of codeword components and tampers according to what he sees. 
Now we first make the intuition that ``an adaptive adversary tampers according to what he sees'' exact. We will take the following approach. Firstly, an adversary is capable of a set of manipulation strategies, each of which is described by a function. Secondly, the adversary is worst-case adversary, namely, he will always apply his best strategy to manipulate. So a code is secure with respect to the adversary means the code is secure against all manipulation strategies the adversary is capable of.  
}

We consider a special type of leakage where the adversary chooses a subset $S, |S| = \rho n $ ($n$ is the codeword length),  and the codeword components associated with this set will be revealed to them. The adversary will then use this information to construct  their
offset vector. 
A tampering strategy is a function from $\mathbb{F}_q^{n}$ to $\mathbb{F}_q^{n}$ which can be described by  the  notation $f_{S,g}$, where $S\subset [n]$  and a function
$g:\mathbb{F}_q^{n\rho}\rightarrow\mathbb{F}_q^{n}$, with the following interpretation. The set S  specifies a  subset of $\rho n$  indexes  of the codeword that the adversary choose. The function $g$ determines an  offset for each read value on the subset $S$.  A \LVAMD code provides protection against all adversary strategies. 
(This approach to defining tampering functions is inspired by Non-Malleable Codes (NMC) \cite{DzPiWi}.)

Let $\mathcal{S}^{[n\rho]}$ be the set of all subsets of $[n]$ of size $n\rho$. Let $\mathcal{M}(\mathbb{F}_q^{n\rho},\mathbb{F}_q^{n})$ denote the set of all functions from $\mathbb{F}_q^{n\rho}$ to $\mathbb{F}_q^{n}$, namely, $\mathcal{M}(\mathbb{F}_q^{n\rho},\mathbb{F}_q^{n}):=\{g:\mathbb{F}_q^{n\rho}\rightarrow\mathbb{F}_q^{n}\}$.

\begin{definition}[$\mathcal{F}^{add}_{\rho}$]
The class of tampering function $\mathcal{F}^{add}_{\rho}$, consists of the set of functions $\mathbb{F}_q^n\rightarrow\mathbb{F}_q^n$, that can be described by two parameters,  $S\in\mathcal{S}^{[n\rho]}$ and $g\in\mathcal{M}(\mathbb{F}_q^{n\rho},\mathbb{F}_q^{n})$.
The set $\mathcal{F}^{add}_{\rho}$ of limited view algebraic tampering functions are defined as follows.
\begin{equation}\label{eq: AMD with leakage}
\mathcal{F}^{add}_{\rho}=\left\{f_{S,g}(\mathbf{x})\ |\ S\in\mathcal{S}^{[n\rho]},g\in\mathcal{M}(\mathbb{F}_q^{n\rho},\mathbb{F}_q^{n})\right\},
\end{equation}
where $f_{S,g}(\mathbf{x})=\mathbf{x}+ g(\mathbf{x}_{|S})$ for $\mathbf{x}\in\mathbb{F}_q^{n}$.
\end{definition}

\begin{definition}[\LVAMD]\label{def: LV-AMD codes} An $(n,k)$-coding scheme is called a strong \LVAMD code with security parameter $\delta$ if $\mbox{Pr}[\mbox{Dec}(f(\mbox{Enc}(\mathbf{m})))\notin\{\mathbf{m},\perp\}]\leq\delta$ for any message $\mathbf{m}\in\mathbb{F}_q^k$ and any $f_{S,g}\in\mathcal{F}^{add}_{\rho}$. It is called a weak \LVAMD code if it only requires the security to hold for a random message $\mathbf{M}\leftarrow\mathbb{F}_q^k$ rather than an arbitrary message $\mathbf{m}$.
\end{definition}


We first give a generic construction of strong \LVAMD codes from WtII codes and AMD codes.
\begin{construction}\label{con: basic construction}
Let $(\mbox{AMDenc},\mbox{AMDdec})$ be a $(q^k,q^{n'},\delta)$-AMD code and let $(\mbox{WtIIenc}, \mbox{WtIIdec})$ be a linear  $(\rho,0)$-wiretap II code with encoder $\mbox{WtIIenc}:\mathbb{F}_q^{n'}\rightarrow\mathbb{F}_q^n$. Then $(\mbox{Enc},\mbox{Dec})$ defined as follows is a strong \LVAMD codes with security parameter $\delta$.
$$
\left\{
\begin{array}{ll}
\mbox{Enc}(\mathbf{m})&=\mbox{WtIIenc}(\mbox{AMDenc}(\mathbf{m}));\\
\mbox{Dec}(\mathbf{x})&=\mbox{AMDdec}(\mbox{WtIIdec}(\mathbf{x})).\\
\end{array}
\right.
$$
When instantiated with the $(q^k,q^{k+2},\frac{k+1}{q})$-AMD code in Construction \ref{ex: AMD} and
the linear  $(\rho,0)$-wiretap II code in Construction \ref{ex: WtII}, we obtain a family of strong \LVAMD codes with security parameter $\frac{k+1}{q}$ that achieves rate $1-\rho$.
\end{construction}


\begin{figure}
\centerline{\includegraphics[scale=0.45]{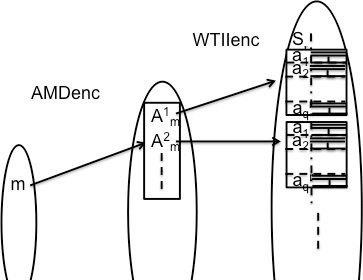}}
\caption{\label{fig: tampering experiment}WtII$\circ$AMD construction with $A^i_\mathbf{m}$ denoting the values of $\mbox{AMDenc}(\mathbf{m})$}
\end{figure}

\begin{proof} 
Since both AMDenc and WtIIenc are randomised encoders, in this proof we write the randomness of a randomized encoder explicitly. Let $I$ denote the randomness of AMDenc and let $J$ denote the randomness of WtIIenc.  
As illustrated in Fig. \ref{fig: tampering experiment}, a message $\bfm$ is first encoded into an AMD codeword $A^I_\bfm=\mbox{AMDenc}(\bfm,I)$. The AMD codeword $A^I_\bfm$ is then further encoded into a WtII codeword, which is the final \LVAMD codeword: $\mbox{Enc}(\bfm)=\mbox{WtIIenc}(A^I_\bfm,J)$.  
According to (\ref{eq: WtII security}), 
$\mbox{SD}\left(\mbox{WtIIenc}(A^{i_1}_\bfm,J)_{|S};\mbox{WtIIenc}(A^{i_2}_\bfm,J)_{|S}\right)=0$. This says that $A^\mathbf{I}_\bfm$ and $\mbox{Enc}(\mathbf{m})_{|S}$ are independent random variables, 
in particular, $\mathbf{I}$ and $\mbox{Enc}(\mathbf{m})_{|S}$ are independent. 
According to Definition \ref{def: LV-AMD codes}, to show that $(\mbox{Enc},\mbox{Dec})$ is a strong \LVAMD code with security parameter $\delta$, we need to show that for any message $\bfm$, and any $f_{S,g}\in\mathcal{F}^{add}_{\rho}$,  $\mbox{Pr}[\mbox{Dec}(f_{S,g}(\mbox{Enc}(\mathbf{m})))\notin \{\mathbf{m},\perp\}] \leq \delta$, where the probability is over the randomness ($\mathbf{I},\mathbf{J}$) of the encoder Enc. We show this in two steps.

\textbf{Step 1.} In this step, we assume that $\mbox{Enc}(\mathbf{m})_{|S}= \bfa$ has occurred and bound the error probability of (Enc,Dec) under this condition. 
We compute
$$
\begin{array}{l}
\mbox{Pr}[\mbox{Dec}(f_{S,g}(\mbox{Enc}(\mathbf{m})))\notin \{\mathbf{m},\perp\}|(\mbox{Enc}(\mathbf{m})_{|S}=\mathbf{a})]\\
=\mbox{Pr}[\mbox{Dec}(\mbox{Enc}(\mathbf{m})+g(\mathbf{a}))\notin \{\mathbf{m},\perp\}|(\mbox{Enc}(\mathbf{m})_{|S}=\mathbf{a})]\\
=\mbox{Pr}[\mbox{AMDdec}(\mbox{WtIIdec}(\mbox{WtIIenc}(\mbox{AMDenc}(\mathbf{m},\mathbf{I}),\mathbf{J})+g(\mathbf{a})))\notin \{\mathbf{m},\perp\}\\
\ \ \ \ |(\mbox{Enc}(\mathbf{m})_{|S}=\mathbf{a})]\\
=\mbox{Pr}[\mbox{AMDdec}(\mbox{AMDenc}(\mathbf{m},\mathbf{I})+\mbox{WtIIdec}(g(\mathbf{a})))\notin \{\mathbf{m},\perp\}|(\mbox{Enc}(\mathbf{m})_{|S}=\mathbf{a})]\\
=\mbox{Pr}[\mbox{AMDdec}(\mbox{AMDenc}(\mathbf{m},\mathbf{I})+\mbox{WtIIdec}(g(\mathbf{a})))\notin \{\mathbf{m},\perp\}],\\
\leq \delta,
\end{array}
$$
where the third equality follows from the linearity of (WtIIenc,WtIIdec), the last equality follows from the fact that $\mathbf{I}$ and $\mbox{Enc}(\mathbf{m})_{|S}$ are independent discussed in the beginning of the proof, 
and the inequality follows trivially from the security of (AMDenc,AMDdec).

\textbf{Step 2.} In this step, we conclude the first part of the proof by showing
$$
\begin{array}{l}
\mbox{Pr}[\mbox{Dec}(f_{S,g}(\mbox{Enc}(\mathbf{m})))\notin \{\mathbf{m},\perp\}]\\
=\sum_{\bfa} \mbox{Pr}[\mbox{Enc}(\mathbf{m})_{|S}= \bfa]\cdot\mbox{Pr}[\mbox{Dec}(f_{S,g}(\mbox{Enc}(\mathbf{m})))\notin \{\mathbf{m},\perp\}|(\mbox{Enc}(\mathbf{m})_{|S}=\mathbf{a})]\\
\leq \sum_{\bfa} \mbox{Pr}[\mbox{Enc}(\mathbf{m})_{|S}= \bfa]\cdot\delta\\
=\delta,
\end{array}
$$
where the inequality follows from \textbf{Step 1.}

Finally, the rate of the $(\rho,0)$-wiretap II code in Construction \ref{ex: WtII} is $\frac{k+2}{n}=1-\rho$. So the asymptotic rate of the strong \LVAMD code family is 
$$
\lim_{n\rightarrow \infty}\frac{k}{n}=\lim_{n\rightarrow \infty}\frac{(1-\rho)n-2}{n}=1-\rho.
$$
\qed
\end{proof}

{
We next show a construction of weak \LVAMD codes that achieves asymptotic rate $1$.
\begin{construction}\label{th: weak LV-AMD}Let $\mathbb{F}_q$ be a finite field of $q$ elements. Let $G$ be a $k\times k$ non-singular matrix over $\mathbb{Z}_{q-1}$ such that each column of $G$ consists of distinct entries, i.e., $g_{i,j}\neq g_{i^{'},j}$ for any $j$ and $i\neq i^{'}$. Assume the entries of $G$ (viewed as integers) is upper-bounded by $\psi k$ for constant $\psi$, i.e., $g_{i,j}\leq \psi k$. Then the following construction gives a family of weak LV-AMD codes of asymptotic rate $1$ with any leakage parameter $\rho<1$. 
$$
\mbox{Enc}\footnote{The message distribution in this construction is not exactly uniform over $\mathbb{F}_q^k$ but $(\mathbb{F}_q^{*})^k$. So this construction can achieve security even when the message distribution is not uniform.}:(\mathbb{F}_q^{*})^k\rightarrow (\mathbb{F}_q^{*})^k\times\mathbb{F}_q: \mathbf{m}\mapsto(\mathbf{m}||f(\mathbf{m},G)),
$$
where $\mathbb{F}_q^{*}$ denotes the set of non-zero elements of $\mathbb{F}_q$ and the tag $f(\mathbf{m},G)$ is generated as follows.
\begin{equation}\label{eq: tag function}
f(\mathbf{m},G)=\sum_{j=1}^k\prod_{i=1}^km_i^{g_{i,j}}.
\end{equation}
The decoder dec checks if the first  $k$-tuple of the input vector, when used in \ref{eq: tag function}, match 
the last component. 
\end{construction}
The proof of Construction \ref{th: weak LV-AMD} is given in Appendix \ref{apdx: proof of theorem weak LV-AMD}.

Concrete constructions of the 
matrix $G$ can be found in 
\cite[Remark 2]{LLR-AMD}. 
}

\section{Applications}\label{sec: applications}
\subsection{Robust Ramp SSS}

A \textit{Secret Sharing Scheme (SSS)}  consists of two algorithms (Share,Recover). The algorithm Share maps a secret $\mathbf{s}\in\mathcal{S}$ to a vector $\mathbf{S}=(S_1,\ldots,S_N)$ where the shares $S_i$ are in some set $\mathcal{S}_i$ and will be 
given to participant $P_i$. The algorithm Recover takes as input a vector of shares $\tilde{\mathbf{S}}=(\tilde{S}_1,\ldots,\tilde{S}_N)$ where $\tilde{S}_i\in\mathcal{S}_i\bigcup\{\perp\}$, where $\perp$ denotes an absent share.  For a $(t, N)$-threshold SSS, $t$ shares reveal no information about the secret $\mathbf{s}$ and $t +1$ shares uniquely recover the secret $\bf s$. 
For a $(t, r, N)$-\textit{ramp} SSS \cite{ramp SSS} with $(\share_\rsss, \recover_\rsss)$ as sharing and recovering algorithms, the access structure is specified by  
two thresholds. The privacy threshold is $t$, and the reconstruction threshold is $r$.
In a $(t, r, N)$-ramp SSS, subsets of $t$ or less shares do not reveal any information about the secret, and subsets of $r$ or more shares can uniquely recover the secret $\bf s$. A set of  shares  of size $t< a < r$ 
may leak some information about the secret. In particular, we consider ramp schemes in which
a set of $t + \alpha (r - t)$ shares leak $\alpha$ fraction of secret  information. \\

\begin{definition}[$(t, r, N)$-Ramp Secret Sharing Scheme] \label{def: rsss} 
A $(t, r, N)$-ramp secret sharing scheme is consist of a \ pair \ of \ algorithms $(\share_\rsss, \recover_\rsss)$, 
where  $\share_\rsss$ randomly maps a secret $\bfs\in\mathcal{S}$ to a share vector $\bfS=(S_1,\cdots,S_N)$ and $\recover_\rsss$ deterministically reconstruct a $\tilde{\bfs}\in\mathcal{S}$ or output $\bot$, satisfy the following.

\begin{itemize}
\item Privacy: The adversary can access up to $r - 1$ shares. If the number of shares accessed by the adversary is $a\leq t$, no information will be leaked about the secret. If the number of leaked share is $a=t + \alpha(r - t)$, where $0 < \alpha < 1$, 
 then
$\tilde{H}_\infty(S|S_{i_1} \cdots S_{i_a})\geq H_\infty(S)-\alpha\log |\mathcal{S}|$  \footnote{This definition of leakage is seemingly different from \cite{new ramp}, where uniform distribution of secret $S$ is assumed and Shannon entropy is used instead of min-entropy. 
}
, for $S\leftarrow\mathcal{S}$ and any $\{i_1, \cdots, i_a\} \subset [N]$.

\item Reconstruction: Any $r$ correct shares can reconstruct the secret $\bfs$.
\end{itemize}
\end{definition}

 A \textit{linear ramp SSS} has the additional property that the Recover function is linear, namely, for any $\mathbf{s}\in\mathcal{G}$, any share vector $\mathbf{S}$ of $\mathbf{s}$, and any vector $\mathbf{S}^{'}$ (possibly containing some $\perp$ symbols), we have $\recover_\rsss( \mathbf{S} + \mathbf{S}^{'}) = \mathbf{s} + \recover_\rsss(\mathbf{S}^{'})$, where vector addition is defined element-wise and addition with $\perp$ is defined by $\perp+  {\bf x} = \mathbf{x} + \perp = \perp$ for all $\mathbf{x}$. 
 In a linear SSS,   the adversary can modify the shares 
  $\tilde{S}_i=S_i+\Delta_i$,  such that   
  the difference $\Delta=\tilde{\mathbf{s}}-\mathbf{s}$ between the reconstructed secret and the shared secret, is known. \\   

In a  $(t, N, \delta)$-robust SSS, for any $t + 1$ shares with at most $t$ shares modified by the adversary, the reconstruction algorithm can recover the secret $\bfs$, or detect the adversarial modification and output $\perp$, with probability at least $1 - \delta$ \cite{AMD}.
That is with probability at most $\delta$ the secret is either not recoverable, or a wrong secret is accepted.
  A modular construction of the robust SSS using an AMD code and a linear SSS  is given by Cramer {\it et al.} \cite{AMD}.\\
 
  We define robust ramp secret sharing scheme when the adversary can adaptively corrupt up to $t + \rho(r - t)$ shares, where $0<\rho<1$ is a constant (level of robustness against active adversaries).

\begin{definition}[$(t, r, N, \rho, \delta)$-Robust Ramp Secret Sharing Scheme] \label{def: rrsss} 
A $(t, r, N, \rho, \delta)$-robust ramp secret sharing scheme is consist of a pair of algorithms $(\share_\rrsss, \recover_\rrsss)$, 
where  $\share_\rrsss$ randomly maps a secret $\bfs\in\mathcal{S}$ to a share vector $\bfS=(S_1,\cdots,S_N)$ and $\recover_\rrsss$ deterministically reconstruct a $\tilde{\bfs}\in\mathcal{S}$ or output $\bot$, satisfy the following.

\begin{itemize}
\item Privacy: The adversary can access up to $r - 1$ shares. If the number of shares accessed by the adversary is $a\leq t$, no information will be leaked about the secret. If the number of leaked share is $a=t + \alpha(r - t)$, where $0 < \alpha < 1$, 
 then
$\tilde{H}_\infty(S|S_{i_1} \cdots S_{i_a})\geq H_\infty(S)-\alpha\log |\mathcal{S}|$, for $S\leftarrow\mathcal{S}$ and any $\{i_1, \cdots, i_a\} \subset [N]$.

\item Reconstruction: Any $r$ correct shares can reconstruct the secret $s$.

\item Robustness: For any $r$ shares with at most $t + \rho(r - t)$ corrupted shares, the probability that either  the secret is correctly reconstructed, or the the adversary's modifications being detected, is at least $1 - \delta$. 

\end{itemize}

\end{definition}

We propose a general construction of robust ramp secret sharing scheme using a  $\rho_\amd$-AMD and  $(t, r, N)$-ramp secret sharing scheme. 

\begin{theorem} Consider a linear $(t, r, N)$-ramp secret sharing scheme with the algorithm pair $(\share_\rsss, \recover_\rsss)$ and shares $\cS_i \in \bF_q^m$,  $i = 1, \cdots, N$, and let (Enc,Dec) be a $\rho_\amd$-AMD code  $\bF_q^k \rightarrow \bF_q^{n}$,  with   failure probability 
 $\delta_\amd$ and $n = (r - t)m$. Then there is a robust ramp secret sharing scheme with algorithm pair  $(\share_\rrsss, \recover_\rrsss)$ given by $\share_\rrsss(\mathbf{s})=\share_\rsss(\mbox{Enc}(\mathbf{s}))$ and $\recover_\rrsss(\tilde{\mathbf{S}})=\mbox{Dec}(\recover_\rsss(\tilde{\mathbf{S}}))$ is a $(t, r, N, \rho, \delta)$-Robust Ramp Secret Sharing Scheme with $\rho \leq \rho_\amd$ and $\delta \leq \delta_\amd$.
\end{theorem}

\begin{proof}
First, we show that if the adversary reads at most $t + \rho(r - t)$ shares, the  $\rho_\amd$-AMD codeword $c$ leaks at most $\rho n \log q$ informations. Since the  $\rho_\amd$-AMD codeword is encoded by a $(t, r, N)$ ramp secret sharing scheme, $t$ shares will not leak any information about the $\rho_\amd$-AMD codeword $c$. Given that the share size  $|\cS_i|\leq q^m$ and $n = (r - t)m$,  
the leakage of the extra $\rho(r - t)$ shares will leak at most $\rho n \log q$ bit of information about the $\rho_\amd$-AMD codeword $c$.

 
Second, we show that the resulting secret sharing scheme is $\delta$-robust. 
For a  secret $\mathbf{s}$, let $\mathbf{S}\leftarrow \share_\rrsss(\mathbf{s})$ be the original  share vector and $\tilde{\mathbf{S}}$ be the corrupted one, and 
let $\mathbf{S}^{'}=\tilde{\mathbf{S}}-\mathbf{S}$. For any $r$ shares, the failure probability of the reconstruction is given by,  
$$
\begin{array}{ll}
\mbox{Pr}[\recover_\rrsss(\tilde{\mathbf{S}}) \notin \{\mathbf{s},\perp\}]&\stackrel{(1)}=\mbox{Pr}[\mbox{Dec}(\recover_\rsss (\mathbf{S})+\recover_\rsss (\mathbf{S}^{'}))\notin\{\mathbf{s},\perp\}]\\
                                                                                                                     &=\mbox{Pr}[\mbox{Dec}(\mbox{Enc}(\mathbf{s})+\Delta)\notin\{\mathbf{s},\perp\}],\\
\end{array}
$$
where $\Delta = \recover_\rsss(\mathbf{S}^{'})$ is chosen by the adversary $\mathbb{A}$, and (1) uses the linearity of the ramp scheme.
In choosing $\Delta$, the adversary  $\mathbb{A}$ can use  at most $\rho$ fraction of  information  
in the  $\rho_\amd$-AMD codeword $c = \mbox{Enc}(\mathbf{s})$.
Since at most $\rho n \log q$ information bit  is leaked to the adversary, that is $\tilde{H}_\infty( C |\bfZ)\geq H_\infty(C) - \rho n \log q $, from the definition of $\rho_\amd$-AMD code with $\rho \leq \rho_\amd$, the  decoding algorithm $\mbox{Dec}$ outputs correct secret $\bfs$, or detects the error $\perp$, with probability at least $1-\delta_\amd$. This means that the ramp secret sharing scheme is robust and outputs either the correct secret, or detects the adversarial tampering, with probability at most  
 $1-\delta \geq 1-\delta_\amd$. 
Thus  a 
$(t, r, N)$-ramp secret sharing scheme and a $\rho_\amd$-AMD with security parameter $\delta_\amd$, give a $(t, r, N, \rho, \delta)$-robust ramp secret sharing scheme with   $\delta \leq \delta_\amd$ and $\rho \leq \rho_\amd$.
\qed
\end{proof}

\subsection{Wiretap II with Algebraic Adversary} 





The Wiretap II \cite{WtII} problem considers a passive adversary that can read a $\rho$ fraction of the codeword components and the goal is to prevent the adversary from 
 learning information about the sent message. 
Wiretap II with an active adversary has been considered in 
\cite{Lai Lifeng} and 
later generalized in \cite{eAWTP,AWTP}. 
In this latter 
 general model, called Adversarial Wiretap (AWTP) mode,  the adversary is characterized by two parameters $\rho_r$ and $\rho_w$, denoting the fraction of the codeword components the adversary can ``read'' and ``modify additively ", respectively.
  The goal 
  is two-fold: to prevent the adversary from obtaining any information (secrecy) and, to recover the message despite the changes made by the 
adversary  
 (reliability). It 
 was proved \cite{AWTP} 
  that in AWTP model, where the adversary can write to a $\rho_w$ fraction of the codeword components additively, secure and
  reliable communication is possible 
  if, $\rho_r+\rho_w <1$. This says that when $\rho_r+\rho_w > 1$, one can only hope for weaker type of security, for example, secrecy and error detection.
  We 
  consider  wiretap II with an algebraic adversary, who can read a $\rho$ fraction of the codeword components and tamper with the whole codeword algebraically, namely, adding a non-zero group element to the codeword (codewords are assumed to live in a group). The adversary in this model is equivalent to the AWTP adversary with $\rho_r=\rho$ and
$\rho_w=1$. 
But the coding goal of wiretap II with an algebraic adversary is different from AWTP.

\begin{definition}\label{def_awtpchannel}
An algebraic tampering wiretap II channel is  
a communication channel between Alice and Bob that is (partially) controlled by an adversary Eve with two following two capabilities.

\begin{itemize}

\item Read:  Eve adaptively selects a fraction $\rho$ of the components 
of the transmitted codeword $\mathbf{c}=c_1,\cdots,c_n$ to read, namely, Eve's knowledge of the transmitted codeword is given by  
$ 
\mathbf{Z}=\{c_{i_1},\cdots, c_{i_{\rho n}}\}
$,
where $S=\{i_1,\cdots,i_{\rho n}\}\subset [n]$ is chosen by Eve.

\item Write:  Assume $\mathbf{c}\in\mathcal{G}$ for some additive group $\mathcal{G}$. Eve chooses an ``off-set'' $\Delta\in\mathcal{G}$ according to $\mathbf{Z}$ and add it to  
the codeword  $\mathbf{c}$, namely, the channel outputs $\mathbf{c}+\Delta$.
\end{itemize}
\end{definition}

\begin{definition}[$(\rho,\epsilon,\delta)$-algebraic tampering wiretap II ($(\rho,\epsilon,\delta)$-AWtII)] \label{def_awtpcode}
 A $(\rho,\epsilon,\delta)$-AWtII code is a coding scheme (Enc,Dec) 
that guarantees the following two properties.
\begin{itemize}
	 
\item {\em Secrecy:} For any pair of messages $\mathbf{m}_0$ and $\mathbf{m}_1$, any $S\subset [n]$ of size $|S|\leq n\rho$, (\ref{eq: WtII security}) should hold, namely,
$$
\SD(\mbox{Enc}(\mathbf{m}_0)_{|S};\mbox{Enc}(\mathbf{m}_1)_{|S})\leq \varepsilon.
$$

\item {\em Robustness:}  If the  adversary is passive, Dec always outputs the correct message. If the adversary is active, the probability that  the decoder outputs a wrong message 
is bounded by $\delta$. That is, for any message $\mathbf{m}$ and any $\rho$-algebraic tampering wiretap II adversary $\mathbb{A}$,
\[
\bP[\eeedec({\mathbb{A}}(\eeeenc(\mathbf{m}))) \notin \{\mathbf{m}, \perp\}] \leq \delta.
\]
\end{itemize}
\end{definition}


The secrecy of $(\rho,\epsilon,\delta)$-AWtII code implies that a $(\rho,\epsilon,\delta)$-AWtII code is a $(\rho,\epsilon)$-WtII code. The following rate upper bound follows directly from Lemma \ref{lem: WtII upper bound}. 
\begin{corollary}
The rate of $(\rho,0,\delta)$-AWtII codes is bounded by $R \leq 1 - \rho$.
\end{corollary}




The robustness property of $(\rho,\epsilon,\delta)$-AWtII code is the same as the security of a strong \LVAMD code (see Definition \ref{def: LV-AMD codes}).
Furthermore, the construction of \LVAMD codes in Construction \ref{con: basic construction} uses a $(\rho, 0)$-WtII code to encode $\mathbf{c}=\mbox{AMDenc}(\mathbf{m})$, which guarantees secrecy with respect to any pair of $(\mathbf{c}_0,\mathbf{c}_1)$, and hence secrecy with respect to any pair of $(\mathbf{m}_0,\mathbf{m}_1)$. These assert that Construction \ref{con: basic construction} yields a family of $(\rho,0,\delta)$-AWtII codes.
\begin{corollary}
There exists a family of $(\rho,0,\delta)$-AWtII codes that achieves rate $R = 1 - \rho$.
\end{corollary}

\section{Conclusion}

We considered an extension of AMD codes when the storage 
leaks information and the amount of leaked information 
 is bounded by $\rho\log|\mathcal{G}|$. 
  We defined \rhoAMD codes that provide protection in this scenario,  both with weak and strong security, and
  derived  concrete and asymptotic bounds on the efficiency of codes in these settings.
Table \ref{tb: bounds} compares our results with original AMD codes and an earlier work (called LLR-AMD)  that allow leakage in specific parts of the encoding process. Unlike LLR-AMD that uses different leakage requirements for the weak and strong case, we use a single model to express the leakage and require that the left-over entropy of the codeword be lower bounded. This makes our analysis and constructions more challenging. In particular,  optimal 
constructions of LLR-AMD codes follow directly from the optimal 
constructions of original AMD codes. However constructing optimal \rhoAMD code,  in both weak and strong model, remain open. We gave an explicit construction of a family of codes with respect to a weaker notion of leakage  (\LVAMD) whose rate achieves the upper bounds of the \rhoAMD codes.
  We finally 
 gave two applications of the codes to robust ramp secret sharing schemes and algebraic manipulation wiretap II channel. 

\appendix
\section*{Appendices}
\addcontentsline{toc}{section}{Appendices}
\renewcommand{\thesubsection}{\Alph{subsection}}

\subsection{Proof of Lemma \ref{lem: strong comparison}}\label{sec: proof of strong comparison}
\begin{proof}

Assume a regular encoder and 
consider a message $\mathbf{m}$.

The codeword $\bfX= \mbox{Enc}(\mathbf{m}, \bfR)$ where the randomness of encoding $\bfR$ is a uniformly distributed $r$-bit string. Now consider an adversary  with leakage variable  $\bfZ$. 
 Because of the one-to-one property of the regular encoder, we have

\begin{eqnarray} \label{eq1}
H_\infty( \bfX) = H_\infty(\bfR)=r,
\end{eqnarray}
and
\begin{eqnarray} \label{eq2}
\begin{array}{ll}
\tilde{H}_\infty( \bfX| \bfZ )&=-\log\mathbb{E}_\mathbf{z}\left(\max_{\mathbf{x}}\mbox{Pr}[\mathbf{X}=\mathbf{x}|\mathbf{Z}=\mathbf{z}]\right)\\
                               &=-\log\mathbb{E}_\mathbf{z}\left(\max_{\mathbf{r}}\mbox{Pr}[\mathbf{R}=\mathbf{r}|\mathbf{Z}=\mathbf{z}]\right)\\
                               &=\tilde{H}_\infty( \bfR| \bfZ ).
\end{array}
\end{eqnarray}

For a leakage variable $\bfZ$, we consider two classes of adversaries denoted by $\mathbb{A}_Z$ and $\mathbb{B}_Z$,  depending on the 
conditions that they must satisfy,  as follows:
$\mathbb{A}_Z()$  is an adversary whose leakage variable 
must satisfy a lower bound on  $\tilde{H}_\infty(\bfR| \bfZ)$  and,
$\mathbb{B}_Z()$ is an adversary  whose leakage variable must satisfy 
a lower bound on $\tilde{H}_\infty(\bfX | \bfZ)$.
Both adversaries, when applied to a vector $x$, use their leakage variables to select an offset vector to be added to a codeword.
That is $\mathbb{A}_Z(x) = x+ \Delta_z$ where $\Delta_z \in F_q^n $  is chosen dependent on the leakage  $\bfZ=z$.  
We have the same definition for  $\mathbb{B}_Z(x) = x+ \Delta_z$.

{\bf  i. strong LLR-AMD code $\Rightarrow$  strong \rhoAMD}\\
Now consider a $(q^k,q^n,2^r,\alpha,\delta)$-strong LLR-AMD code $\bf C$ with encoder and decoder pair,  (Enc,Dec).  
For an adversary $\mathbb{A}_Z$ whose  leakage  variable 
satisfies  $ \tilde{H}_\infty(\bfR |\bfZ) \geq (1-\alpha)r$,  we have

$$
\mbox{Pr}[ \mbox{Dec}  (  \mathbb{A}_Z ( \mbox{Enc} ( \mathbf{m},\mathbf{R} ) ) ) 
\notin\{\mathbf{m},\perp\}]   \leq\delta,
$$
where the probability is over the randomness of encoding, and is an expectation over $\mathbf{z}\in\mathcal{Z}$. 

Note that using (\ref{eq1}) and  (\ref{eq2}),   the  $\mathbb{A}_Z$ adversary is also a $\mathbb{B}_Z$ adversary 
satisfying,

\begin{eqnarray} \label{eq3}
\tilde{H}_\infty(\bfX |\bfZ)\geq \tilde{H}_\infty( \bfX) -\alpha r
\end{eqnarray}

Both these adversaries have the same leakage variable $\bfZ$ and so any algorithm Offset$(z)$  used by one,  taking  the value $\bfZ=z$ as input
and  finding the 
the offset $\Delta_z$, can be used by the other also (the two adversaries have the same information). This means that
%
the success probabilities of the two adversaries are the same,
$$
\mbox{Pr}[\mbox{Dec} (\mathbb{A}_Z (\mbox{Enc} (\mathbf{m},\mathbf{R}) ) )
\notin\{\mathbf{m},\perp\}]=
\mbox{Pr}[\mbox{Dec}(\mathbb{B}_Z  (\mbox{Enc}(\mathbf{m},\mathbf{R}) )) 
\notin\{\mathbf{m},\perp\}] \leq\delta.
$$



For \rhoAMD codes, security  is defined against  a $\mathbb{B}_Z$ adversary whose leakage variable $\bfZ$ satisfies,

\begin{eqnarray} \label{eq4}
\tilde{H}_\infty(\bfX |\bfZ)\geq H_\infty( \bfX) -\rho n\log q
\end{eqnarray}
Comparing  (\ref{eq4}) and (\ref{eq3}),  we conclude that  $\bf C$ is a \rhoAMD code for $\rho$  values that satisfy  $\alpha r  \geq \rho n\log q$, namely $\rho\leq\frac{\alpha r}{n\log q}$. 


{\bf  ii. strong  \rhoAMD  $\Rightarrow$  strong LLR-AMD code\\}
An argument similar to \textbf{i.} immediately gives that the $(q^k,q^n,2^r,\alpha,\delta)$-strong LLR-AMD code obtain from \rhoAMD code should satisfy $\alpha\leq\frac{\rho n\log q}{ r}$. Next we show the bound on $r$ follows from Proposition \ref{prop1} together with (\ref{eq1}). Indeed, by Proposition \ref{prop1}, $\tilde{H}_\infty(\bfX |\bfZ)\geq\log\frac{1}{\delta}$ should hold for any $\mathbf{Z}$ satisfying $\tilde{H}_\infty(\bfX |\bfZ)\geq H_\infty( \bfX) -\rho n\log q$. In particular, we must have $H_\infty( \bfX) -\rho n\log q\geq\log\frac{1}{\delta}$. Now we can use (\ref{eq1}) to conclude that $r\geq\log\frac{1}{\delta}+\rho n\log q$.
\qed\end{proof}

\subsection{Proof of Lemma \ref{lem: weak comparison}}\label{sec: proof of weak comparison}
\begin{proof}
The encoder Enc is a one-to-one correspondence between messages and codewords. 
Consider a message variable $\mathbf{M}\leftarrow\mathcal{M}$ (in particular, the uniform distribution is emphasized by $\mathbf{M}_u\stackrel{\$}{\leftarrow}\mathcal{M}$).
The codeword is a variable $\bfX= \mbox{Enc}(\mathbf{M})$. Now consider an adversary  with leakage variable  $\bfZ$. 
Because of the one-to-one property of the encoder, we have

\begin{eqnarray} \label{eq: 1}
H_\infty( \bfX) = H_\infty(\bfM),
\end{eqnarray}
and
\begin{eqnarray} \label{eq: 2}
\begin{array}{ll}
\tilde{H}_\infty( \bfX| \bfZ )&=-\log\mathbb{E}_\mathbf{z}\left(\max_{\mathbf{x}}\mbox{Pr}[\mathbf{X}=\mathbf{x}|\mathbf{Z}=\mathbf{z}]\right)\\
                               &=-\log\mathbb{E}_\mathbf{z}\left(\max_{\mathbf{m}}\mbox{Pr}[\mathbf{M}=\mathbf{m}|\mathbf{Z}=\mathbf{z}]\right)\\
                               &=\tilde{H}_\infty( \bfM| \bfZ ).
\end{array}
\end{eqnarray}

For a leakage variable $\bfZ$, we consider two classes of adversaries denoted by $\mathbb{A}_Z$ and $\mathbb{B}_Z$,  depending on the 
conditions that they must satisfy,  as follows:
$\mathbb{A}_Z()$  is an adversary whose leakage variable 
must satisfy a lower bound on  $\tilde{H}_\infty(\bfM| \bfZ)$  and,
$\mathbb{B}_Z()$ is an adversary  whose leakage variable must satisfy 
a lower bound on $\tilde{H}_\infty(\bfX | \bfZ)$.
Both adversaries, when applied to a vector $x$, use their leakage variables to select an offset vector to be added to a codeword.
That is $\mathbb{A}_Z(x) = x+ \Delta_z$ where $\Delta_z \in F_q^n $  is chosen dependent on the leakage  $\bfZ=z$.  
We have the same definition for  $\mathbb{B}_Z(x) = x+ \Delta_z$.

{\bf  i. weak LLR-AMD code $\Rightarrow$  weak \rhoAMD}\\
Now consider a $(q^k,q^n,\alpha,\delta)$-weak LLR-AMD code $\bf C$ with encoder and decoder pair,  (Enc,Dec).  
For an adversary $\mathbb{A}_Z$ whose  leakage  variable 
satisfies  $ \tilde{H}_\infty(\bfM |\bfZ) \geq (1-\alpha)k\log q$,  we have

$$
\mbox{Pr}[ \mbox{Dec}  (  \mathbb{A}_Z ( \mbox{Enc} ( \mathbf{M}) ) ) 
\notin\{\mathbf{M},\perp\}]   \leq\delta,
$$
where the probability is over the randomness of encoding, and is an expectation over $\mathbf{z}\in\mathcal{Z}$. 

Note that using (\ref{eq: 1}) and  (\ref{eq: 2}),   the  $\mathbb{A}_Z$ adversary is also a $\mathbb{B}_Z$ adversary 
satisfying,

\begin{eqnarray} \label{eq: 3}
\tilde{H}_\infty(\bfX |\bfZ)\geq (1-\alpha)k\log q
\end{eqnarray}

Both these adversaries have the same leakage variable $\bfZ$ and so any algorithm Offset$(z)$  used by one,  taking  the value $\bfZ=z$ as input
and  finding the 
the offset $\Delta_z$, can be used by the other also (the two adversaries have the same information). This means that
%
the success probabilities of the two adversaries are the same,
$$
\mbox{Pr}[\mbox{Dec} (\mathbb{A}_Z (\mbox{Enc} (\mathbf{M}) ) )
\notin\{\mathbf{M},\perp\}]=
\mbox{Pr}[\mbox{Dec}(\mathbb{B}_Z  (\mbox{Enc}(\mathbf{M}_u) )) 
\notin\{\mathbf{M}_u,\perp\}] \leq\delta.
$$



For \rhoAMD codes, security  is defined against  a $\mathbb{B}_Z$ adversary whose leakage variable $\bfZ$ satisfies,

\begin{eqnarray} \label{eq: 4}
\tilde{H}_\infty(\bfX |\bfZ)\geq H_\infty( \bfX) -\rho n\log q,\mbox{ where } \bfX=\mbox{Enc}(\bfM_u).
\end{eqnarray}
Comparing  (\ref{eq: 4}) and (\ref{eq: 3}),  we conclude that  $\bf C$ is a \rhoAMD code for $\rho$  values that satisfy  $\alpha k  \geq \rho n$, namely $\rho\leq\frac{\alpha k}{n}$. 


 {\bf  ii. weak  \rhoAMD  $\Rightarrow$  weak LLR-AMD code\\}
An argument similar to \textbf{i.} immediately gives that the $(q^k,q^n,\alpha,\delta)$-weak LLR-AMD code obtain from \rhoAMD code should satisfy $\alpha\leq\frac{\rho n}{ k}$. 
\qed\end{proof}

\subsection{Proof of Proposition \ref{pro: weak bound}}\label{sec: proof of weak bound}

\begin{proof} 
By Proposition \ref{prop1}, $\tilde{H}_\infty(\bfX |\bfZ)\geq\log\frac{1}{\delta}$ should hold for any $\mathbf{Z}$ satisfying $\tilde{H}_\infty(\bfX |\bfZ)\geq H_\infty( \bfX) -\rho n\log q$. In particular, we must have $H_\infty( \bfX) -\rho n\log q\geq\log\frac{1}{\delta}$. Since the message $\bfM$ of weak \rhoAMD is uniform and the encoder is one-to-one correspondence, $H_\infty( \bfX)=H_\infty( \bfM)=k\log q$. We conclude that $k\log q-\rho n\log q\geq\log\frac{1}{\delta}$, namely, 
\begin{equation}\label{eq: weak rho AMD 1}
q^{\rho n-k}\leq \delta.
\end{equation} 
Similar to the proof of Theorem \ref{lem: generalised upper bound}, we also consider a random attack strategy. Then the total number of valid codewords that do not decode to $\mathbf{M}$ is at least $(q^{k}-1)$, which  is the number of offsets that lead to undetected manipulations. A randomly chosen offset ($\Delta\neq 0^n$) leads to undetected manipulation with probability at most
 $$
 \frac{q^{k}-1}{q^{n}-1}
 $$
 and we must have
\begin{equation}\label{eq: weak rho AMD 2}
\frac{q^{k}-1}{q^{n}-1}\leq \delta.
\end{equation}
\qed\end{proof}
\subsection{Proof of Construction \ref{th: weak LV-AMD}} \label{apdx: proof of theorem weak LV-AMD}

\begin{proof}  Let $\beta$ be a primitive element of $\mathbb{F}_q$. Then every element $m_i\in\mathbb{F}_q^{*}$ can be written as a power of $\beta$: $m_i=\beta^{m_i^{'}}$. (\ref{eq: tag function}) is rewritten as follows.
$$
f(\mathbf{m},G)=\sum_{j=1}^k\beta^{\sum_{i=1}^km_i^{'}g_{i,j}\mbox{ mod }(q-1)}.
$$
According to \cite[Theorem 4]{LLR-AMD} and the proof therein, (Enc,Dec) satisfies $\mbox{Pr}[\mbox{Dec}(\mbox{Enc}(\mathbf{m})+\Delta(\mathbf{Z}_\rho))\notin\{\mathbf{m},\perp\}]\leq\frac{\psi k}{q-1}$ as long as the leakage parameter $\rho$ satisfies $k-(k+1)\rho\geq 1$. What is left to show is for any $\rho<1$ and $\delta>0$, there exists an $N$ such that for all $k+1\geq N$, $k-(k+1)\rho>0$ and $\frac{\psi k}{q-1}\leq \delta$ are both satisfied. Indeed, $k-(k+1)\rho=k(1-\rho)-\rho$, which is bigger than $1$ if $k>\frac{1+\rho}{1-\rho}$. So we can simply let $N=\lceil\frac{1+\rho}{1-\rho}\rceil+1$. And $\frac{\psi k}{q-1}\leq \delta$ can be achieved by choosing a big enough $q$, for example, $q=\omega(\psi k)$ and choose a big enough $k$. 
\qed\end{proof}



\end{document}